\documentclass[letterpaper]{article} 
\usepackage{aaai2026}  
\usepackage{times}  
\usepackage{helvet}  
\usepackage{courier}  
\usepackage[hyphens]{url}  
\usepackage{graphicx} 
\usepackage{natbib}  
\usepackage{caption} 
\usepackage{algorithm}
\usepackage{algorithmic}

\usepackage{booktabs, amsfonts, nicefrac, microtype, amsmath} 

\usepackage{color, amssymb, amsthm, bigints, graphicx, subcaption, pifont, float}
\usepackage{mathrsfs}
\usepackage{multirow}

\theoremstyle{definition}
\newtheorem{definition}{Definition}
\newtheorem{proposition}{Proposition}
\newtheorem{assumption}{Assumption}

\usepackage{newfloat}
\usepackage{listings}
\DeclareCaptionStyle{ruled}{labelfont=normalfont,labelsep=colon,strut=off} 
\floatstyle{ruled}
\newfloat{listing}{tb}{lst}{}
\floatname{listing}{Listing}
\title{LILAD: Learning In-context Lyapunov-stable Adaptive Dynamics Models}
\author{
    Amit Jena\textsuperscript{\rm 1},
    Na Li\textsuperscript{\rm 2},
    Le Xie\textsuperscript{\rm 2}
}
\affiliations{
    \textsuperscript{\rm 1} Department of Electrical and Computer Engineering, Texas A\&M University, USA \\
    \textsuperscript{\rm 2} Harvard John A. Paulson School of Engineering and Applied Sciences, Harvard University, USA \\

    amit\_jena@tamu.edu, nali@seas.harvard.edu, xie@seas.harvard.edu
}

\usepackage{bibentry}

\begin{document}

\maketitle

\begin{abstract}
System identification in control theory aims to approximate dynamical systems from trajectory data. While neural networks have demonstrated strong predictive accuracy, they often fail to preserve critical physical properties such as stability and typically assume stationary dynamics, limiting their applicability under distribution shifts. Existing approaches generally address either stability or adaptability in isolation, lacking a unified framework that ensures both.
We propose LILAD (Learning In-Context Lyapunov-stable Adaptive Dynamics), a novel framework for system identification that jointly guarantees adaptability and stability. LILAD simultaneously learns a dynamics model and a Lyapunov function through in-context learning (ICL), explicitly accounting for parametric uncertainty. Trained across a diverse set of tasks, LILAD produces a stability-aware, adaptive dynamics model alongside an adaptive Lyapunov certificate. At test time, both components adapt to a new system instance using a short trajectory prompt, which enables fast generalization.
To rigorously ensure stability, LILAD also computes a state-dependent attenuator that enforces a sufficient decrease condition on the Lyapunov function for any state in the new system instance. This mechanism extends stability guarantees even under out-of-distribution and out-of-task scenarios.
We evaluate LILAD on benchmark autonomous systems and demonstrate that it outperforms adaptive, robust, and non-adaptive baselines in predictive accuracy.
\end{abstract}

\section{Introduction}
Identifying dynamical systems from trajectory data is a foundational problem with wide-ranging applications in power systems, biological systems, and beyond. Despite the extensive literature on linear system identification, many real-world systems are inherently nonlinear and lie beyond the scope of linear models. Moreover, learning a linear approximation from data generated by a nonlinear system introduces modeling errors, often resulting in suboptimal performance. Many research works attempt to alleviate this challenge by employing nonlinear basis functions, such as sinusoidal series \cite{korenberg1989sinusoid} or sum-of-squares representations \cite{li2018sos}, as well as iterative optimization methods like Monte Carlo techniques \cite{schon2015montecarlo} and genetic algorithms \cite{chang2007genetic}. However, these approaches rely on restrictive assumptions such as polynomial structure or sufficient smoothness of dynamics. Such assumptions limit their applicability, and may lead to a failure in capturing the complex, high-frequency transient behavior exhibited by many real-world systems.

Neural networks offer a flexible system identification framework due to their expressive representational capacity, enabling modeling of a broader class of systems. However, due to their black box nature, the key physical properties of the underlying dynamics such as stability, monotonicity, dissipativity don't reflect in the learned models. Standard training on state–next-state pairs typically prioritizes accuracy over preserving these essential system characteristics, thereby inhibiting the deployability of such learned models in safety-critical applications.

A recent body of work has focused on embedding physical properties into learned models, aiming to create trustworthy digital twins of the underlying systems. One such approach \cite{kolter2019learning} enforces stability by jointly training a dynamics model with a neural Lyapunov function. Another notable method \cite{kang2021stable} induces Lyapunov stability in neural ordinary differential equations (ODEs), making the resulting models robust to adversarial perturbations. Despite their strong empirical performance and theoretical foundations, these approaches suffer from a key limitation: the assumption of stationarity in the underlying system. However, many practical systems don't comply with this assumption, exhibiting time-varying dynamics or stochasticity in system parameters. For example, the topology of a power distribution network is often altered during the system operation for different objectives such as fault isolation or load balancing, leading to a time-varying dynamics. Another example is the seasonal variability in solar irradiance, which introduces parametric uncertainty in PV-inverter-integrated power systems. 
Among all physical properties, \textit{stability} holds the most critical importance for ensuring reliable system operation. However, existing works overlook the possibility of non-stationary dynamics. A common remedy like retraining a stable dynamics model from scratch whenever the system changes, requires a new training dataset and significant additional training time, rendering it impractical for time-sensitive operations. Another parallel line of research employs adaptive learning methods such as meta-learning \cite{park2022meta} or in-context learning \cite{forgione2023icl} to model non-stationary dynamics under mild assumptions. However, similar to standard neural network approaches, these methods aim to minimize squared loss errors without ensuring the stability of the learned dynamics.
As a result, \textit{learning stability-preserving dynamics models that remain effective under mild constraints in non-stationary settings remains an open and challenging problem}. 

In this paper, we study the problem of learning stability-preserving adaptive dynamics models. We propose \textbf{LILAD} (Learning In-context Lyapunov-stable Adaptive Dynamics Models), a novel framework that leverages the in-context learning (ICL) capability of a large language model (LLM) to approximate nonstationary dynamics, while enforcing stability through Lyapunov-based conditions. We assume parametric stochasticity in the system, with access to a data pool comprising multi-task trajectory datasets. Each dataset is generated by instantiating the system with a specific parameter sample. During the training of LILAD, we employ the ICL procedure to jointly train a dynamics model and a Lyapunov function that satisfies global exponential stability across the multi-task data pool. Specifically, we adopt an adversarial training framework that alternates between updating the dynamics model and the Lyapunov function to promote joint and efficient learning. At evaluation, LILAD receives a small test-time trajectory dataset as a prompt corresponding to a new, unseen parameter sample, and produces a stable dynamics in a zero-shot manner. Furthermore, using a bisection-based root-finding method, we compute an attenuator $\gamma$ using the trained Lyapunov and dynamics models, which strictly enforces stability in the learned dynamics model. We validate LILAD against baseline methods on several benchmark autonomous systems, including a high-dimensional example to demonstrate the adaptability and stability performance of our approach. Our codes are available at https://github.com/amitjena1992/LILAD.

\section{Related Work}
\noindent \textbf{System Identification:} System Identification encompasses a broad range of algorithms, all aimed at approximating models that accurately predict the future state of a system given its current state. Traditional modeling approaches include linear models \cite{mehra1974linear}, as well as representation using basis functions like sinusoidal series \cite{korenberg1989sinusoid}, sum-of-squares polynomials \cite{li2018sos},  kernels like quadratic and Gaussian functions \cite{greblicki2003quadgauss}. Another class of works focuses on block-oriented methods, including Wiener models \cite{hagenblad1999weiner}, Hammerstein models \cite{chaoui2005hammerstein}, and hybrid Wiener–Hammerstein structures \cite{sun2022robust}, which combine linear blocks and static nonlinearity to approximate nonlinear systems. In parallel, a different line of study utilizes iterative algorithms such as genetic algorithms \cite{chang2007genetic} and Monte Carlo methods \cite{schon2015montecarlo}, as well as Bayesian approaches like Gaussian process regression \cite{bijl2017GPR} to approximate nonlinear dynamics in a data-efficient and uncertainty-aware manner. More recent advances advocate for neural network-based techniques that are highly expressive and free from restrictive model assumptions. Examples include convolutional neural networks (CNNs) \cite{wu2019cnn}, Neural ODEs \cite{quaglino2019NODE}, and physics-informed neural networks (PINNs) \cite{stiasny2021PINN}.\\

\noindent\textbf{Stability-preserving System Identification:}
Identifying stable systems is of utmost priority, particularly when the learned model is intended for deployment in a safety-critical system. These concerns have laid the foundation for a dedicated section of work. In this setting, various methods have been proposed that rigorously guarantee stability in the approximated dynamics. Classical approaches include stable spline kernels \cite{pillonetto2010spline}, Lyapunov-based frameworks \cite{lyashevskiy1994nonlinear}, subspace methods \cite{maciejowski1995subspace}, convex parameterizations for stable systems \cite{tobenkin2010convex}, and approximation of stable positive systems \cite{umenberger2016positive}. With the advent of deep learning, this line of research has advanced with learning-based techniques like Lyapunov-based neural network methods \cite{kolter2019learning, lawrence2020almost}, stable neural ODEs \cite{kang2021stable}, stability-aware deep generative models \cite{urain2020imitationflow}, and Lyapunov-stable deep equilibrium models \cite{chu2024equilbrium}.\\

\noindent\textbf{Adaptive and Robust System Identification:}
The presence of stochasticity in a dynamical system necessitates adaptive and robust methods for reliable operation and control. This challenge has been addressed through both classical and modern approaches. Classical adaptive linear methods include least mean square (LMS)-based techniques \cite{chen2009LMS}, adaptive Wiener models \cite{ogunfunmi2007adaptive}, and robust error-domain model falsification (EDMF) methods \cite{pasquier2015robust}. More recently, learning-based approaches have gained traction, including meta-learning \cite{park2022meta} and in-context learning (ICL) \cite{forgione2023icl, jena2025llm}, which offer greater flexibility in handling complex, nonstationary dynamics.

Despite progress in learning-based stable and adaptive modeling separately, no prior work has unified both objectives. To the best of our knowledge, this is the first neural network-based approach to address stable adaptive system identification, demonstrating improved performance over standard adaptive baselines.
\section{Preliminaries}
\subsection{Stability of Discrete-Time Dynamical Systems}
We summarize the fundamental stability results for a generic discrete-time deterministic system of the form:
\begin{equation}
\label{eq:gen_dyn}
    x_{k+1} = f(x_k),
\end{equation}
where $x_k \in X \subset \mathbb{R}^d $ denotes the system state at discrete time $k$, and $f(.): \mathbb{R}^d \rightarrow \mathbb{R}^d$ represents the dynamics. The equilibrium point of the system \eqref{eq:gen_dyn} is assumed to be $x = 0$.

\begin{definition}[Global exponential Stability (GES) \cite{grujic1974exponential}]
\label{def:GES}
The equilibrium $x = 0$ is said to be globally exponentially stable if and only if there exist constants $m  > 1$, $\alpha >0$ such that
\begin{equation}
\label{eq:gen_exp_stab}
    \|x_k\|_2 \leq m\|x_0\|_2 e^{-\alpha k}, \forall k \in \mathbb{N},
\end{equation}
for any initial state $x_0 \in X$.
\end{definition}
\noindent GES implies global asymptotic stability where $x_k \rightarrow 0$ as $k \rightarrow \infty$. Additionally, it indicates 
that the rate of decay is exponential in time with decay rate $\alpha$. 
\begin{definition}[Lyapunov functions for discrete-time systems 
\cite{grujic1974exponential}]
\label{def:Lyap}
    The equilibrium $x = 0$ is globally exponentially stable for system \eqref{eq:gen_dyn} if and only if there exists a scalar-valued function $V : \mathbb{R}^d \rightarrow \mathbb{R}$ with the following properties:
    \begin{subequations}
    \label{cond:Lyap}
    \begin{align}
        &V(x)\text{ is continuous in } \mathbb{R}^d, \label{cond:Lyap1}\\
        & c_1\|x\|_2^2 \leq V(x) \leq c_2\|x\|_2^2, \quad \forall x \in X, \label{cond:Lyap2}\\
        & V(x_{k+1}) \leq \beta V(x_k),\quad \forall x_k \in X,~~ \forall k \in \{0,1,\hdots\} \label{cond:Lyap3}
        \end{align} 
    \end{subequations}
    where $c_1 > 0$, $c_2 > 0$,  $0 < \beta < 1$.
\end{definition}
\noindent Condition $\eqref{cond:Lyap2}$ implies that $V(x) > 0$ for all $x \in X \backslash \{0\}$ and $V(x) = 0$ if and only if $x = 0$, consistent with other variants of Lyapunov-based definitions, such as those in \cite{bof2018lyapunov}. Additionally, many existing works \cite{lawrence2020almost} conventionally consider $\beta = 1- \alpha$ where $\alpha \in (0,1)$ is the decay rate of the dynamics as defined in Definition \ref{def:GES}. The existence of a Lyapunov function satisfying the conditions in Definition \ref{def:Lyap} ensures that the system \eqref{eq:gen_dyn} is globally exponentially stable, as stated in Definition \ref{def:GES}. 

\subsection{In-context Learning}
We briefly summarize the key principles of ICL, the underlying basis of our proposed LILAD framework. 
Let $\mathcal{T}$ be a family of tasks, and let $\mathcal{P}_{\tau}$ denote a meta-distribution across $\mathcal{T}$. Any task $\tau \in \mathcal{T}$ induces a task distribution $P_{Z|\tau}$ on the domain $Z = U \times V$, where $U$ and $V$ are the input space and output space, respectively. Given a random prompt
\begin{equation}
    \hspace{-0.5em} \mathscr{P}^{\tau} = \{u_1^\tau, v_1^\tau, u_2^\tau, v_2^\tau, \hdots, u_n^\tau, v_n^\tau, u_{n+1}^\tau\}
\end{equation}
which comprises $n$ in-context input-output pairs and a query input $u_{n+1}^\tau$, where each $(u_i^\tau, v_i^\tau) \sim P_{Z|\tau}$ i.i.d. for $i \in \{1, \hdots, n\}$. Additionally, the query input and true output $(u_{n+1}^\tau, v_{n+1}^\tau) \sim P_{Z|\tau}$ are also drawn independently from the same distribution. In this setting, an ICL model $M(.)$ aims to predict $v_{n+1}^\tau$ with an average error bounded by $\varepsilon$. Formally,
\begin{equation}
    \label{ICL}
    \mathbb{E}_{P_\tau}\mathbb{E}_{{P_{Z|\tau}}^{n+1}}\big[ \ell(M(\mathscr{P}^{\tau}), v_{n+1}^\tau)\big] \leq \varepsilon,
\end{equation}
where $\ell(.,.)$ is an appropriate loss function. In this process, the model $M(\cdot)$ adapts to any task $\tau$ solely through the prompt $\mathscr{P}_\tau$, without requiring any gradient updates. 

\section{ICL of Adaptive and Stable Dynamics}
We consider a discrete-time autonomous system of the form
\begin{equation}
    \label{eq:dis_dyn}
    x_{k+1} = f_{\vartheta}(x_k),
\end{equation}
where $x_k \in X \subset \mathbb{R}^d$ denotes the fully observed time-evolving state of the system, $f_{\vartheta} : \mathbb{R}^d \rightarrow \mathbb{R}^d$ represents the parametrized dynamics, and $\vartheta \in \mathbb{R}^\mathfrak{d}$ is a stochastic system parameter governed by an unknown probability distribution $\mathcal{P}_{\vartheta}$. Although the distribution $\mathcal{P}_{\vartheta}$ is explicitly unknown, its variability is experienced through a finite set of parameter samples $\mathcal{D}^{\vartheta} := \{\tilde\vartheta_1, \tilde\vartheta_2, \hdots, \tilde\vartheta_M\}$, where each $\tilde\vartheta_i$ leads to a distinct deterministic dynamics $f_{\tilde\vartheta_i}(.)$. We assume that the origin is the equilibrium point of each such dynamical system, i.e., $f_{\tilde\vartheta_i}(0) = 0$ holds for every $\tilde\vartheta_i \in \mathcal{D}^{\vartheta}$. Next, we make the following assumption regarding the stability of the dynamics, which is fundamental to our formulation.
\begin{assumption}[GES]
\label{assum:stab_dyn}
Let $\tilde\vartheta_i \in \mathcal{D}^{\vartheta}$ be a parameter sample, and let $f_{\tilde\vartheta_i}(.)$ be its corresponding induced dynamics. Then for some $m>1$ and $\alpha >0$, the dynamics $x_{k+1} = f_{\tilde \vartheta_i}(x_k)$ satisfies GES condition as defined in Definition \ref{def:GES}.
\end{assumption}
\noindent Assumption \ref{assum:stab_dyn} ensures that any induced dynamics $f_{\tilde\vartheta_i}(.)$ is globally exponentially stable, meaning its trajectories converge to the equilibrium at an exponential rate for any initial condition. This type of assumption is widely used in the stability analysis literature \cite{grujic1974exponential, geiselhart2014alternative} in various forms. Our problem formulation is developed under this assumption.

Approximating the nonstationary dynamics as stated in \eqref{eq:dis_dyn} from trajectory data under parametric uncertainty is challenging, especially when stability must be preserved. We address this by proposing an adaptive ICL-based method that ensures stable dynamics. In our set-up, we assume access to a training data pool $\mathcal{D}^\text{train} := \{\mathcal{D}^{i}\}_{i=1}^{M} = \{\{(x_{i,j}, f_{\tilde \vartheta_i}(x_{i,j}))\}_{j=1}^N\}_{i=1}^M$, consisting of $M$ datasets. Each dataset $\mathcal{D}^{i}$ corresponds to a distinct dynamics $f_{\tilde\vartheta_i}(.)$ induced by $\tilde \vartheta_i \in \mathcal{D}^\vartheta$, and contains $N$ state–next-state pairs, generated from rollouts initialized at random initial states. We note that the state-next-state pairs have been randomly shuffled to remove the temporal dependencies between samples, resulting in i.i.d. datasets suitable for standard learning algorithms. Thus, $x_{i,j}$ and $x_{i,j+1}$ are not sequential states but are independently sampled from the dataset.

We present LILAD, an ICL-based framework that jointly learns adaptive dynamics and Lyapunov models. LILAD training proceeds by alternating between two phases: (i) updating the Lyapunov model using a \textit{Lyapunov loss} computed with the current fixed dynamics model, and (ii) updating the dynamics model using a combination of standard MSE and a \textit{dynamics correction loss} derived from the fixed Lyapunov model. This process continues until both models converge, resulting in a stability-aware dynamics model and a corresponding Lyapunov function.

Both models jointly satisfy the Lyapunov stability constraints (as defined in Definition \ref{def:Lyap}) for most samples in the training data pool. However, like many neural network-based methods, the \textit{out-of-task} and \textit{out-of-sample} performances of LILAD remain uncertain. We note that due to the design of our Lyapunov model architecture, the first two conditions \eqref{cond:Lyap1} and \eqref{cond:Lyap2} are inherently satisfied. Consequently, we focus on ensuring the satisfaction of the final condition, \eqref{cond:Lyap3}, which enforces exponential stability. To this end, we utilize a bisection-based root-finding method to compute an attenuation factor $\gamma \in [0,1)$ that guarantees Lyapunov decrease across out-of-task and out-of-sample scenarios. The key idea is to sufficiently attenuate the dynamics model to ensure a suitable decrease in the Lyapunov value between two consecutive states, thereby satisfying the exponential stability condition. It is important to note that $\gamma$ is not computed globally; instead, it is determined separately for each state, making it inherently state-dependent. In what follows, first we describe the LILAD set-up in detail, and then proceed to the training and evaluation procedures of LILAD. 

\subsection{LILAD Framework Set-up}
We begin by introducing the following notations for our adaptive system identification framework:
\begin{subequations}
    \begin{align}
        & \mathscr{P}^i  := \{x_{i,1}, f_{\tilde \vartheta_i}(x_{i, 1}), \hdots, x_{i,n}, f_{\tilde\vartheta_i}(x_{i, n}), x_{i, n+1}\} \\
        & \mathscr{P}_j^i  := \{x_{i,1}, f_{\tilde\vartheta_i}(x_{i, 1}), \hdots, x_{i,j}, f_{\tilde\vartheta_i}(x_{i, j}), x_{i, j+1}\}, \\
        & \mathscr{C}^i = \mathscr{P}^i \setminus \{x_{i, n+1}\}, \qquad \mathscr{C}_j^i = \mathscr{P}_j^i \setminus \{x_{i, j+1}\},
    \end{align}
\end{subequations}
where $\mathscr{P}^i$ denotes the full prompt for the $i$-th task, consisting of $n$ in-context state–next-state pairs $\mathscr{C}^i$ and a query $x_{i,n+1}$. This prompt is constructed by randomly selecting $n$ pairs from the $N$ available state–next-state pairs associated with $i$-th task. Similarly, $\mathscr{P}_j^i$ is the $(i,j)$-th prompt prefix, comprising the first $j$ in-context pairs $\mathscr{C}_j^i$ and query $x_{i,j+1}$. We refer to $\mathscr{C}^i$ and $\mathscr{C}_j^i$ as $i$-th context and $(i,j)$-th context prefix, respectively. In this setting, the proposed ICL-based dynamics model $G_\theta$ takes any prompt prefix $\mathscr{P}_j^i$ as input and predicts $f_{\tilde \vartheta_i}(x_{i, j+1})$. In parallel, the ICL-based Lyapunov model $V_\phi$ also receives $\mathscr{P}_j^i$ and outputs the Lyapunov value corresponding to $x_{i, j+1}$. We present these predictions as follows:
\begin{subequations}
\label{eq:lilad_pred}
\begin{align}
    &\hat{f}(x_{i, j+1}) := G_\theta(\mathscr{P}_j^i) = G_\theta(x_{i,j+1} ~|~ \mathscr{C}_j^i), \\
    & \hat{V}(x_{i, j+1}) := V_\phi(\mathscr{P}_j^i) = V_\phi(x_{i,j+1} ~|~ \mathscr{C}_j^i),
\end{align}
\end{subequations}
where $\hat{f}(x_{i, j+1})$ and $\hat{V}(x_{i, j+1})$ denote the predicted dynamics and Lyapunov values at $x_{i,j+1}$ conditioned on the context $\mathscr{C}_j^i$. We note that the context prefix $\mathscr{C}_0^i$ is empty for all cases. However, this does not inhibit prediction, as the ICL models can still generate outputs based on the query input alone, without any contextual information. 

The architectures for $G_\theta$ and $V_\phi$ follow a common structure comprising a linear input layer, a GPT-2 transformer block, and a linear output layer. In both models, the input layer tokenizes the prompt and passes it to the transformer block, which generates a context-aware embedding for a prediction on the query. Then, this prediction is passed through the output layer to produce the final model output. While standard transformer architectures ensure that the first Lyapunov constraint \eqref{cond:Lyap1} holds any general query conditioned on a fixed context, the second condition \eqref{cond:Lyap2} isn't inherently satisfied. We make the following architectural modifications to overcome this challenge.
\subsubsection{Enforcing Semi-definiteness via Output Warping}
We employ a scaled $\tanh$ activation after the output layer to constrain the Lyapunov output within the range $[-c, c]$ for some $c > 0$, which improves the numerical tractability of integrating $V_\phi$ into the training pipeline. To further enforce the semi-definiteness of $V_\phi$, we adopt a smoothed ReLU activation inspired by \cite{kolter2019learning}, defined as:
\begin{equation}
    \sigma(x)= \begin{cases}
    0 & \text { if } x \leq 0 \\
    x^2 / 2 \delta & \text { if } 0<x<\delta \\
    x-\delta / 2 & \text { otherwise },
    \end{cases}
\end{equation}
where $\delta > 0$ is a tunable parameter regulating the smoothness of the function. Then, the following output warping scheme is applied to ensure that $V_\phi$ remains positive semi-definite:
\begin{align}
    V_\phi(x_{i,j+1} ~|~ \mathscr{C}_j^i) = ~& \sigma\big(c\cdot\tanh(V^\text{raw}_\phi(x_{i,j+1} ~|~ \mathscr{C}_j^i))   \nonumber \\
       -~ c\cdot\tanh&(V^\text{raw}_\phi(0 ~|~ \mathscr{C}_j^i))\big)  + \epsilon~\|x_{i, j+1}\|^2,
    \end{align}
where $V^\text{raw}_\phi(. ~|~ \mathscr{C}_j^i)$ denotes unwarped Lyapunov output, and $\epsilon > 0$ ensures a unique minima at $x_{i, j+1} = 0$.
Now, it's straightforward to prove that with the positive semi-definiteness and boundedness constraints imposed, $ V_\phi(x_{i,j+1} ~|~ \mathscr{C}_j^i)$ satisfies the second Lyapunov condition \eqref{cond:Lyap2} for any $x_{i, j+1} \in X$. In the remainder of this section, we separately describe the training and evaluation procedures for LILAD.
\subsection{LILAD Training}
We adopt an adversarial training procedure consisting of two stages. In the first stage, the dynamics model $G_\theta$ is frozen while $V_\phi$ is trained. In the second phase, $V_\phi$ is fixed and $G_\theta$ is updated. This alternating process progresses until both models converge in terms of their respective loss functions. The dynamics and Lyapunov losses, represented by $\mathcal{L}^\text{Dyn}(\theta)$ and $\mathcal{L}^\text{Lyap}(\phi)$, are defined as follows:
\begin{subequations}
    \label{eq:loss_funs}
    \begin{align}
        & \mathcal{L}^\text{Dyn}(\theta)  = \dfrac{1}{M(n+1)} \sum_{i=1}^M \sum_{j=0}^n  \|G_\theta(x_{i, j+1}|\mathscr{C}_j^i)\nonumber\\ 
         & - f_{\tilde\vartheta_i}(x_{i, j+1})\| + \lambda \max\big\{ V_\phi\big(G_\theta(x_{i, j+1}|\mathscr{C}_j^i) \big|\mathscr{C}_j^i\big)\nonumber \\
         & \hspace{10em} - \beta V_\phi(x_{i, j+1} | \mathscr{C}_j^i ),~~0\big\}, \label{eq:loss_funs_dyn}\\
        &\mathcal{L}^\text{Lyap}(\phi) = \dfrac{1}{M(n+1)} \sum_{i=1}^M \sum_{j=0}^n  \max\big\{ \nonumber \\
        &  \hspace{1em}  V_\phi\big(G_\theta(x_{i, j+1}|\mathscr{C}_j^i) \big|\mathscr{C}_j^i\big) - \beta V_\phi(x_{i, j+1} | \mathscr{C}_j^i ) ,~~ 0\big\}, \label{eq:loss_funs_lyap}
    \end{align}
\end{subequations}
where $\lambda > 0$ is a small regularization constant, and $\beta$ is the exponential stability coefficient, as defined in \eqref{cond:Lyap3}. For the dynamics loss $\mathcal{L}^\text{Dyn}(\theta)$, the first term inside the summation is a standard MSE between the context-conditioned predicted next state and the true next state. The second term penalizes violations of the Lyapunov condition \eqref{cond:Lyap3}, serving as a dynamics correction term while $V_\phi$ is held fixed. Similarly, for $\mathcal{L}^\text{Lyap}(\phi)$, the dynamics model $G_\theta$ is frozen and the loss captures the extent to which the current $V_\phi$ violates the Lyapunov condition \eqref{cond:Lyap3}. Thus, both models contribute to the evaluation of both loss functions, enabling a joint and efficient training process. Upon convergence, the training algorithm yields a stability-aware dynamics model $G_{\theta^{*}}$ and a Lyapunov function $V_{\phi^*}$, both optimized with respect to their respective loss functions. The complete LILAD training procedure is summarized in Algorithm \ref{algor:LILAD_train}. 

\subsection{LILAD Evaluation and Stability Enforcement}
The optimal models $G_{\theta^{*}}$ and $V_{\phi^*}$ obtained from LILAD training jointly satisfy Lyapunov constraints for the majority of samples in the training prompts $\{\mathscr{P}_i\}_{i=1}^M$. However, deploying such models in safety-critical systems demands rigorous stability guarantees, which go beyond empirical performance and require stronger, verifiable assurances. We address this bottleneck by presenting a bisection-based root-finding method to compute a state-dependent attenuator $\gamma$, which rigorously enforces Lyapunov stability for arbitrary states, thereby generalizing to both out-of-task and out-of-sample scenarios. Our approach builds on the technique introduced in \cite{lawrence2020almost}, which was originally developed for stationary dynamics. We extend this idea to the non-stationary setting considered in this work. We note that, contrary to many prior works \cite{kolter2019learning, chu2024equilbrium}, our method does not assume input-output convexity of the Lyapunov model $V_{\phi}$, a property that is difficult to verify for transformer-based LLMs.

The strategy followed in this step is to first check if the context-conditioned next-state prediction by the optimal dynamics model satisfies the third Lyapunov condition \eqref{cond:Lyap3}. Specifically, at any arbitrary query state $x \in X$, given a new test-time context prefix $\mathscr{C}_j^{M+1}$ corresponding to a new $(M+1)$-th task, we verify if $V_{\phi^*}\big(G_{\theta^{*}}(x|\mathscr{C}_j^{M+1}) \big|\mathscr{C}_j^{M+1}\big) - \beta V_{\phi^*}(x | \mathscr{C}_j^{M+1}) <0$. If this condition is satisfied at $x$, we set $\gamma = 1$. Otherwise, we compute a suitable $\gamma$ as follows.
\begin{align}
    \label{eq:gamma_compute}
    &\text{Find the largest}~~\gamma \in [0,1)~~\text{such that} \nonumber\\
    &\hspace{-0.5em}V_\phi^{*}\big(\gamma\cdot G_\theta^{*}(x|\mathscr{C}_j^{M+1}) \big|\mathscr{C}_j^{M+1}\big) - \beta V_\phi^{*}(x | \mathscr{C}_j^{M+1} ) = 0.
\end{align}
This results in a one-dimensional root-finding problem, solvable with bisection, that ensures exponential stability of the modulated dynamics $\gamma \cdot G_\theta^{*}$. We utilize the intermediate value theorem (IVT) to guarantee the existence of a valid $\gamma$ satisfying \eqref{eq:gamma_compute}. The detailed proof of this claim has been deferred to the Appendix. It is important to note that $\gamma$ is state-dependent, i.e., it varies with each query state and is not shared across tasks or samples. Hence, we denote it as $\gamma(x)$. Finally, the stability-constrained optimal dynamics model is represented as $\gamma(x)\cdot G_\theta^{*}(x | \mathscr{C}_j^{M+1})$, where
\begin{equation}
    \gamma(x) = 
    \begin{cases}
     & \hspace{-1em}1 \hspace{2.5em}\text{if}~~ V_{\phi^{*}}\big(G_{\theta^*}(x~|~\mathscr{C}_j^{M+1}) \big|\mathscr{C}_j^{M+1}\big) \\
     &\hspace{6em}- \beta V_{\phi^{*}}(x | \mathscr{C}_j^{M+1} ) <0 \\
    & \hspace{-1em}\text{solution of \eqref{eq:gamma_compute}} \quad \text{otherwise}.
    \end{cases}
\end{equation}
\begin{algorithm}[t]
    \begin{algorithmic}[1]
\caption{ LILAD Training Algorithm}
\label{algor:LILAD_train}
    \STATE {\textbf{Input:}} Offline prompts $\{\mathscr{P}_i\}_{i=1}^{M}$, with each $\mathscr{P}_i$ comprising $n$ state-next-state pairs; switching interval $K^\text{sw}$. 
    \STATE {\textbf{Initialize:}} Dynamics and Lyapunov model parameters  $\theta^{(0)}, \phi^{(0)}$ 
   \FOR{$k = 0, \hdots, K-1$}
   \STATE Randomly sample a batch of prompt prefixes $\{\mathscr{P}^{p}_i\}_{i=1}^M$, where $p < n$.
   \IF{$\lfloor(k / K^\text{sw})\rfloor \bmod 2= 0$} 
        \STATE Compute and update dynamics loss $\mathcal{L}^{\text{Dyn}}(\theta^{(k)})$ as described in \eqref{eq:loss_funs_dyn} over $\{\mathscr{P}^{p}_i\}_{i=1}^M$ while fixing $V_{\phi^{(k)}}$. 
    \ELSE
        \STATE Compute and update Lyapunov loss $\mathcal{L}^{\text{Lyap}}(\phi^{(k)})$ as described in \eqref{eq:loss_funs_lyap} over $\{\mathscr{P}^{p}_i\}_{i=1}^M$ while freezing $G_{\theta^{(k)}}$.
    \ENDIF
   \ENDFOR
  \STATE {\textbf{Output:}} $\theta^{*} = \theta^{(K)}$, $\phi^{*} = \phi^{(K)}$  
\end{algorithmic}
\end{algorithm}

\section{Experiments}

\begin{figure*}[t]
    \vspace{-1em}
     \centering
     \begin{subfigure}[b]{0.32\textwidth}
         \includegraphics[width=\textwidth]{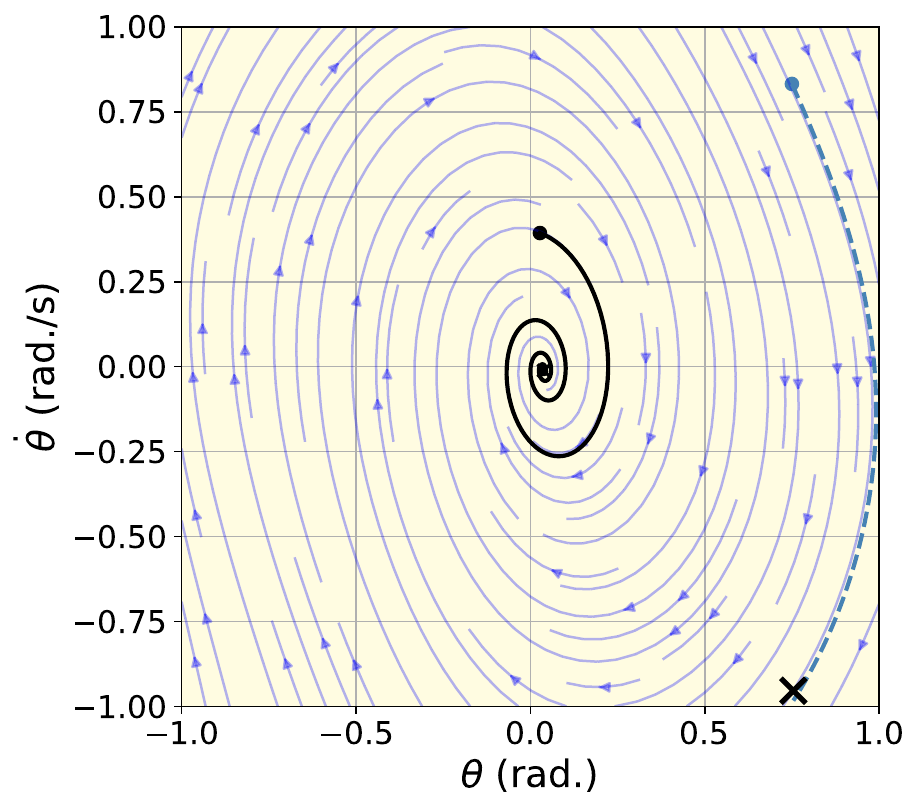}
         \caption{ICL}
         \label{fig:icl_pp}
     \end{subfigure}
     \hfill
     \begin{subfigure}[b]{0.32\textwidth}
         \includegraphics[width=\textwidth]{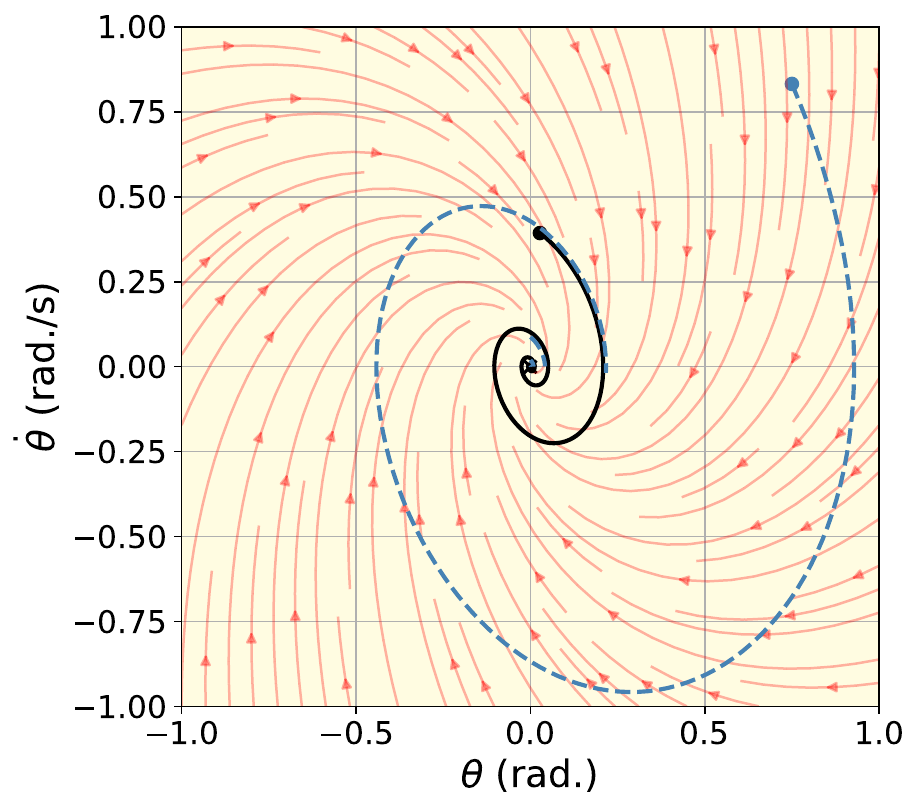}
         \caption{LILAD}
         \label{fig:lilad_pp}
     \end{subfigure}
     \hfill
     \begin{subfigure}[b]{0.32\textwidth}
         \includegraphics[width=\textwidth]{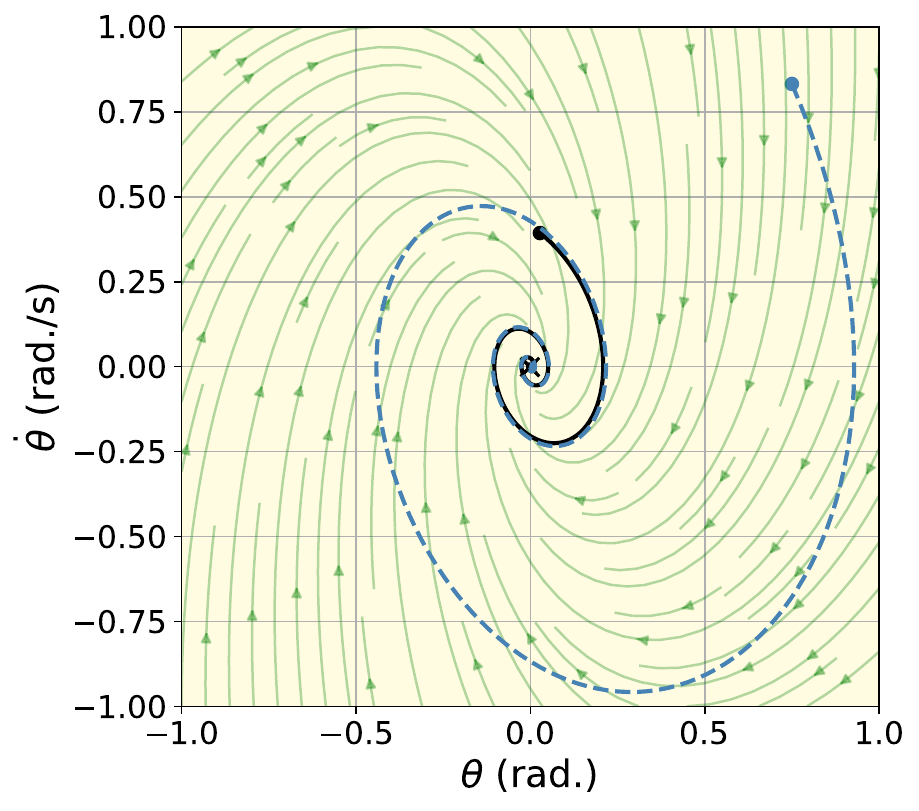}
         \caption{Ground Truth}
         \label{fig:gt_pp}
     \end{subfigure}
     \caption{\textit{Performance Comparison of LILAD}. For the Simple Pendulum system with stochastic parameters $(g, l, b)$, the proposed LILAD method adapts more effectively to test-time instances than ICL. (a) ICL yields a sub-optimal approximation of the test-time vector field, causing predicted trajectories to deviate beyond the specified boundary. (b) In contrast, LILAD achieves a more accurate approximation, with guaranteed convergence of the predicted trajectories to the equilibrium. (c) Ground truth: the true vector field and corresponding reference trajectories of the test-time system.}
     \label{fig:pendulum_2d}
\end{figure*}
\begin{figure}[h]
    \centering
    \includegraphics[width=1\linewidth]{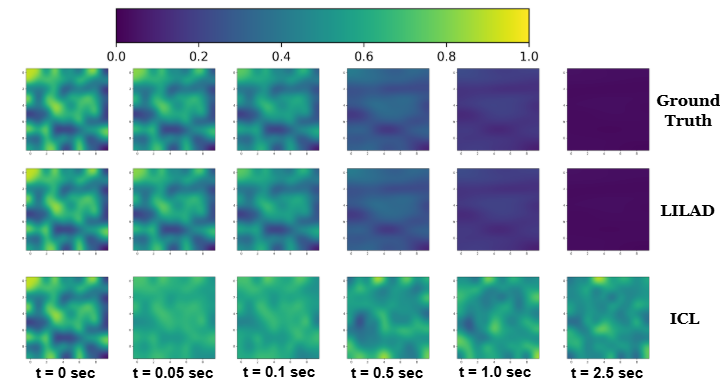}
    \caption{LILAD scales effectively to a high-dimensional system derived from a discretized reaction-diffusion PDE. Its learned surrogate predicts trajectories that converge to the origin, matching the test-time ground truth. In contrast, ICL does not guarantee convergence to the origin. Each image shows the spatial temperature profile under a fixed test-time diffusion coefficient.}
    \label{fig:LILAD-high-dim}
    \vspace{-2em}
\end{figure}
\begin{table*}[h]
\caption{Comparison of MAE and RMSE of all Methods Across Dynamical Systems}
\label{tab:performance}
\vspace{-1em}
\begin{center}
\begin{small}
\begin{sc}
\begin{tabular}{lclccccr}
\hline
Env. & dim & Metric & ICL & MAML & CVaR & Stable-Linear & LILAD\\
\hline
\multirow{2}{*}{SP} 
& \multirow{2}{*}{2} & MAE   & 0.018 $\pm$ 0.068 & 0.023$\pm$0.014 & 0.085 $\pm$ 0.022 & 0.065$\pm$0.016 & \textbf{0.004 $\pm$ 0.020} \\
&                    & RMSE & 0.022 $\pm$ 0.080 & 0.034 $\pm$ 0.024  & 0.17 $\pm$ 0.09 & 0.13$\pm$0.021 &  \textbf{0.011 $\pm$ 0.012} \\
\hline
\multirow{2}{*}{DP} 
& \multirow{2}{*}{4} & MAE  & 0.039 $\pm$ 0.029 & 0.022 $\pm$ 0.010 & 0.12 $\pm$ 0.05 & 0.17 $\pm$ 0.16 & \textbf{0.011 $\pm$ 0.008} \\
&                   & RMSE & 0.061$\pm$ 0.028 & 0.032 $\pm$ 0.008 & 0.22 $\pm$ 0.12 & 0.42$\pm$0.25 &  \textbf{0.021 $\pm$ 0.018} \\
\hline
\multirow{2}{*}{MG} 
& \multirow{2}{*}{5} & MAE  & 0.005 $\pm$ 0.010 & 0.007 $\pm$ 0.009 & 0.014 $\pm$ 0.028 & 0.011 $\pm$ 0.004 & \textbf{0.005 $\pm$ 0.002} \\
&                    & RMSE & \textbf{0.009 $\pm$ 0.012} & 0.009 $\pm$ 0.013 &  0.042 $\pm$ 0.068 & 0.017 $\pm$ 0.004 &  0.012 $\pm$ 0.009\\
\hline
\multirow{2}{*}{SEIR} 
& \multirow{2}{*}{8} & MAE  & 0.022 $\pm$ 0.017 & 0.032 $\pm$ 0.027 & 0.077 $\pm$ 0.05 & 1.049 $\pm$ 0.818 & \textbf{0.017 $\pm$ 0.019} \\
&                    & RMSE & 0.031 $\pm$ 0.011 & 0.043 $\pm$ 0.067 & 0.18 $\pm$ 0.13 & 1.617 $\pm$ 1.404 & \textbf{0.022 $\pm$ 0.029} \\
\hline
\multirow{2}{*}{PDE-SM} 
& \multirow{2}{*}{100} & MAE  & 6.354 $\pm$ 0.050 & \textendash & \textendash & \textendash & \textbf{0.060 $\pm$ 0.002} \\
&                    & RMSE & 7.478 $\pm$ 0.039 & \textendash & \textendash & \textendash & \textbf{0.249 $\pm$ 0.008} \\
\hline
\end{tabular}
\end{sc}
\end{small}
\end{center}
\vspace{-1em}
\end{table*}

The effectiveness of our proposed LILAD is demonstrated through comparisons with standard baselines across multiple benchmark systems. Specifically, we consider five autonomous systems, each of which is inherently stable and either does not require external control or is already equipped with a stabilizing controller, such as pre-defined droop controllers in the case of the networked microgrid. As explained in earlier sections, we focus on the system identification task under parametric uncertainty. The benchmark systems used in our simulations include: (i) a one-link pendulum, (ii) a two-link pendulum, (iii) a networked microgrid, (iv) an SEIR epidemic model, and (v)  a high-dimensional surrogate model of a reaction-diffusion partial differential equation (PDE).
\subsection{Baselines}
We consider four baseline methods for performance comparison. Our choice of baselines includes two adaptive algorithms, one robust algorithm, and a non-adaptive stable system identification method. The adaptive baselines consist of standard ICL and model-agnostic meta-learning (MAML) \cite{finn2017MAML}, both of which have been widely used in adaptive regression and non-linear system identification tasks \cite{forgione2023icl, park2022meta}. In case of ICL, after a full course of in-context training, the ICL model makes predictions on a general query state given a test-time context prefix $\mathscr{C}_j^{M+1}$, similar to our proposed LILAD framework. In contrast, MAML does not depend on the ordering of $\mathscr{C}_{j}^{M+1}$; it treats the test-time context as an unordered set of examples and performs gradient-based fine-tuning to adapt to new tasks. For the robust baseline, we adopt a Conditional Value-at-Risk (CVaR)-based regression method \cite{rockafellar2000cvar}, which employs a risk-sensitive max-min formulation to learn a system identification model that performs well under worst-case scenarios. Specifically, we follow a neural network-based approach that, instead of minimizing the average loss, identifies the top $k$ worst-case losses, and optimizes over their average value. This method doesn't require a test-time prompt as it directly optimizes in a risk-aware manner during training, yielding a fixed model that generalizes to uncertain or adverse conditions. Finally, we include a linear stable system identification approach (referred to as stable-linear) that enforces stability by constraining the eigenvalues of the learned state-transition matrix to lie in the unit circle. We note that this method serves as a non-adaptive and non-robust baseline, as it is directly computed on the test-time context prefix. Finally, under parametric uncertainty, each method (including LILAD) is evaluated on its ability to adapt to new test-time tasks while maintaining stability in the identified system. 

\subsection{Simulations}

We use mean absolute error (MAE) and root mean squared error (RMSE) as performance metrics to assess the efficacy of each method. Specifically, we evaluate a system identification method by first rolling out trajectories using the model adapted to the test-time task and then computing the MAE and RMSE between these predicted trajectories and the ground-truth trajectories generated from the actual test-time dynamics. To account for parametric uncertainty during evaluation, we instantiate five test-time systems by drawing five random parameter samples. For each of these test-time systems, we fix four initial states, which are used to generate both the ground truth trajectories and predicted trajectories. Finally, for each system, we compute a total of twenty MAE and RMSE values, one for each trajectory. The mean and standard deviation of these MAE and RMSE values for each system are reported in Table \ref{tab:performance}. In this comparison table, we use shorthands SP, DP, MG, SEIR, PDE-SM to denote the simple pendulum, double pendulum, networked microgrids, SEIR, and the PDE surrogate system, respectively.

We consider a simple pendulum setting as our first system, and assume the pendulum's length ($l$), damping coefficient ($b$), and acceleration due to gravity ($g$) to be stochastic parameters. The system states are angular displacement from the vertically downward position $\theta$ and angular velocity $\dot{\theta}$. We compute the MAE and RMSE metrics for each method and observe that LILAD outperforms all baselines, including the stable-linear method that utilizes the test-time data directly, and is already at an advantage. Quantitatively, LILAD obtains $4.5 \times$ lower mean MAE and $2 \times$ lower mean RMSE compared to ICL, the next best-performing approach. We further compare LILAD and ICL visually in Fig. \ref{fig:pendulum_2d}, which shows the phase plots of a representative test-time system instance for LILAD, ICL, and the ground-truth dynamics. Each plot includes two rollout trajectories from the same initial points. We notice that LILAD more accurately captures the underlying vector field of the actual dynamics compared to ICL. Furthermore, both trajectories generated by LILAD stay within the bounding box as the ground truth trajectories, whereas ICL results in a trajectory that deviates beyond these bounds. This indicates that LILAD provides superior approximation quality, not only offering better accuracy but also better approximating the system's vector field. This gives our method a strong advantage for deployment in safety-critical settings. Among the rest of the baselines, all significantly perform worse than both ICL and LILAD, with the exception of MAML, which exceeds ICL but underperforms compared to LILAD.

Our second benchmark is  a double pendulum setting, where the two masses $m_1$ and $m_2$, the length of two links $l_1$ and $l_2$, damping coefficients $b_1$ and $b_2$ are assumed to take stochastic values. The states are $\theta_1$, $\theta_2$, $\dot\theta_1$, and $\dot \theta_2$, that represent the angular displacements and velocities. We evaluate all methods using the MAE and RMSE and observe that LILAD consistently supersedes all baselines in both metrics, mimicking its strong performance in the simple pendulum system. Specifically, LILAD fetches $2 \times$ lower mean MAE and $1.5 \times$ lower mean RMSE than MAML, which is the next best method for this setting. The ICL falls behind both MAML and LILAD, but yields substantially better results than CVaR and stable-linear approaches.

The third benchmark is networked-microgrid environment consisting of five microgrids that interact with each other through power electronic interfaces at the points of common coupling (PCCs). Each microgrid employs a voltage angle droop controller to stabilize its local subsystem. The overall dynamics are governed by the droop control in conjunction with AC power flow constraints. The state comprises the modified voltage phase angles $[\Delta\delta_1, \hdots, \Delta\delta_5]$, where each $\Delta \delta_j = \delta_j - \delta_{j}^{*}$ denotes the difference between the $j$-th phase angle and its setpoint (i.e., equilibrium value). We assume that the droop constants $[D_{\delta_1}, \hdots, D_{\delta_5}]$ to be stochastic parameters. We evaluate all methods in terms of MAE and RMSE and find that LILAD matches ICL in mean MAE and slightly lags behind ICL in mean RMSE. This behavior is attributed to LILAD's more accurate approximation of the steady-state part, despite a minor under-approximation of the transient part. Both LILAD and ICL outperform the remaining baselines in this setting.

The fourth system models the progress of a disease outbreak in a population. It divides the general population into four groups: susceptible (S), exposed (E), infectious (I), and removed (R). We consider a scenario involving two interacting populations undergoing active vaccination campaigns. The state vector includes $S_1, E_1, I_1, R_1, S_2, E_2, I_2, R_2$ that indicate the population density in each compartment across both populations. The self transmission rates $(\beta_{11}$ , $\beta_{22})$, and recovery rates $\gamma_1$ and $\gamma_2$ are treated as stochastic parameters. Based on the computed MAE and RMSE values for all methods, we observe that LILAD outperforms all baselines, demonstrating its superior adaptability to diverse and heterogeneous environments. In sharp contrast, both CVaR and stable-linear provide the weakest performance, with stable-linear in particular severely under-approximating the system dynamics, resulting in the poorest metrics.

\subsubsection{Scalability of LILAD to High-dimensional Systems}
To demonstrate the scalability of LILAD, we take a stable reaction-diffusion PDE that governs the spatiotemporal evolution of temperature across a surface. The surface is discretized into a $10\times10$ grid, yielding 100 nodal temperatures and resulting in a dynamical system of dimension 100. We treat the diffusion coefficient $\alpha_\text{diff}$ as stochastic and evaluate LILAD and ICL across five test-time instances. Consistent with previous results, we find that ICL doesn't always guarantee the convergence of the adapted model to the origin during test time. In contrast, the predicted trajectories by LILAD always reach equilibrium, thereby closely matching the test-time ground truth. We demonstrate this behavior in Figure \ref{fig:LILAD-high-dim}, where a LILAD-based representative trajectory is visually compared and contrasted with ICL's diverging prediction using a sequence of heatmaps. Furthermore, as reported in Table \ref{tab:performance}, LILAD achieves approximately $106$ times lower mean MAE and $30$ times lower RMSE than ICL, suggesting the much higher reliability of LILAD in high-dimensional non-stationary safety-critical applications. In this setting, we exclude MAML, CVaR, and stable-linear baselines due to their lack of scalability and consistently subpar performance on large-scale systems. 
\section{Conclusion}
In this work, we proposed LILAD, a novel system identification framework that unifies stability and adaptability, two properties typically studied in isolation. Utilizing in-context learning (ICL), LILAD employs two transformer models to jointly learn a stability-aware adaptive dynamics model and a valid adaptive Lyapunov function from a multi-task trajectory pool. During evaluation, both models adapt seamlessly to new dynamics using only contextual information in the form of test-time prompts. To further ensure stability, we incorporated a state-dependent attenuator $\gamma$, which guarantees sufficient Lyapunov decrease across predicted state transitions. We validated LILAD’s effectiveness across benchmark  systems, demonstrating superior performance over adaptive, robust, and non-adaptive baselines. To the best of our knowledge, LILAD is the first neural network-based framework to achieve adaptive system identification while guaranteeing stability across diverse tasks.

In future work, we aim to extend LILAD to state estimation and particle filtering problems under parametric uncertainty, as well as explore adaptive control through ICL from multi-task data.
\section*{Acknowledgements}
This work was supported in part by the NSF AI Institute (Award No. 2112085) and the Salata Institute for Climate and Sustainability at Harvard University.

The views expressed in this paper are the opinion of the authors and do not reflect the views of PJM Interconnection, L.L.C. or its Board of Managers of which Le Xie is a member.

\bibliography{aaai2026}
\onecolumn
\appendix
\section{On the Existence of a Solution for $\gamma$ in Problem~\eqref{eq:gamma_compute}}
In the subsection \textit{LILAD Evaluation and Stability Enforcement}, we informally state that a solution of the attenuation factor $\gamma$ is guaranteed to exist in $[0,1]$. This is a founding result that enables the computation of a suitable $\gamma$ for problem \ref{eq:gamma_compute}. 

\begin{proposition}
   Let $G_{\theta^{*}}$ and $V_{\phi^{*}}$ be the optimal adaptive dynamics and Lyapunov models obtained from LILAD training, respectively. Let $\mathscr{C}_j^{M+1}$ be a test-time prompt-prefix containing $j$ state-next-state pairs corresponding to a new test-time system instance $x_{k+1} = f_{\tilde\vartheta_{M+1}}(x_k)$. Let the exponential Lyapunov constraint be violated at an arbitrary $x \in X$, i.e., $V_{\phi^{*}}\big(G_{\theta^*}(x~|~\mathscr{C}_j^{M+1}) \big|\mathscr{C}_j^{M+1}\big) - \beta V_{\phi^{*}}(x | \mathscr{C}_j^{M+1} ) > 0$.  
   Then, a solution of $\gamma(x)$ exists in $[0,1]$ for the following root-finding problem:
   \begin{equation}
   \label{eq:gamma_exist_original}
       V_\phi^{*}\big(\gamma(x)\cdot G_\theta^{*}(x|\mathscr{C}_j^{M+1}) \big|\mathscr{C}_j^{M+1}\big) - \beta V_\phi^{*}(x | \mathscr{C}_j^{M+1} ) = 0
   \end{equation}
\end{proposition}
\begin{proof}
    We use an intermediate value theorem (IVT) argument to establish this result. First, at any state $x \in X$, we construct the following function
    \begin{equation*}
    H(\gamma_x) : =    V_\phi^{*}\big(\gamma_x \cdot G_{\theta^{*}}(x|\mathscr{C}_j^{M+1}) \big|\mathscr{C}_j^{M+1}\big) - \beta V_{\phi^{*}}(x | \mathscr{C}_j^{M+1} ),  
    \end{equation*}
    where $\gamma_x= \gamma(x)$, assumed for simplicity of exposition. As $V_{\phi^*}$ follows a standard transformer architecture, $H(\gamma_x)$ is continuous for $\gamma_x \in [0,1]$. Next, because of the violation of the Lyapunov constraint at $x$, the following is true:
    \begin{eqnarray}
    \label{eq:gamma_exit_upper}
        &V_{\phi^{*}}\big(G_{\theta^{*}}(x|\mathscr{C}_j^{M+1}) \big|\mathscr{C}_j^{M+1}\big) - \beta V_{\phi^{*}}(x | \mathscr{C}_j^{M+1}) > 0 \nonumber \\
        \equiv & H(1) > 0 ,
    \end{eqnarray}
\noindent where $H(1) = H(\gamma_x)\big|_{\gamma_x = 1}$. Next, because of the output warping technique-based architectural modification, we have the following for any parameter $\phi$:
\begin{equation*}
    V_\phi(x ~|~ \mathscr{C}_j^{M+1}) = \sigma\big(c\cdot\tanh(V^\text{raw}_\phi(x ~|~ \mathscr{C}_j^{M+1})) -~ c\cdot\tanh(V^\text{raw}_\phi(0 ~|~ \mathscr{C}_j^{M+1}))\big)  + \epsilon~\|x\|^2,
\end{equation*}
which directly leads to:
\begin{subequations}
\begin{align}
    & V_{\phi^{*}}(x ~|~ \mathscr{C}_j^{M+1})  >  0 \hspace{2em} \text{if }~ x \in X \setminus \{0\} \label{eq:semidef_1}\\
     & V_{\phi^{*}}(0 ~|~ \mathscr{C}_j^{M+1})  =  0 \label{eq:semidef_2}
\end{align}
\end{subequations}

\noindent We use these results to establish the following for $H(0) = H(\gamma_x)\big|_{\gamma_x = 0}$:
\begin{eqnarray}
 \label{eq:gamma_exit_lower}
     & \hspace{-20em}H(0) \nonumber \\
    =&  \hspace{-2.5em}V_{\phi^{*}}\big(0 \cdot G_{\theta^{*}}(x|\mathscr{C}_j^{M+1}) \big|\mathscr{C}_j^{M+1}\big) - \beta V_{\phi^{*}}(x | \mathscr{C}_j^{M+1} ) \nonumber \\
    =&\hspace{-9em}  V_{\phi^{*}}\big(0 |\mathscr{C}_j^{M+1}\big) - \beta V_{\phi^{*}}(x | \mathscr{C}_j^{M+1} ) \nonumber\\
   =& 0 - \beta V_{\phi^{*}}(x | \mathscr{C}_j^{M+1} ) \nonumber \hspace{8em}(\because \text{from} \eqref{eq:semidef_2}),\\
   = & - \beta V_{\phi^{*}}(x | \mathscr{C}_j^{M+1} ) < 0\hspace{8em}(\because \text{from} \eqref{eq:semidef_1}). 
\end{eqnarray}
Finally, using the continuity of $H(\gamma_x)$, and considering the results in \eqref{eq:gamma_exit_upper} and  \eqref{eq:gamma_exit_lower}, we invoke IVT and state that a solution of $H(\gamma_x)$ exists for $\gamma_x \in [0,1]$ because of the change of signs at two extreme points. Due to the construction of $H(\gamma_x)$, this directly leads to the existence of a solution for Problem \ref{eq:gamma_exist_original}. 
\end{proof}
\section{Description of Computational Resources}
All experiments have been conducted in a Dell XPS 8920 desktop with Core i7‑7700 quad‑core @3.6 GHz (boost up to 4.2 GHz), 8 MB cache, 65 W TDP. Total NVIDIA GPU memory and CPU memory used for training are 96 GB and 256 GB, respectively. We have utilized Python-based Numpy and Pytorch packages to train LILAD, ICL and all other neural network-based baseline methods. For stable-linear method, we have implemented Numpy and cvxpy python modules.
\section{Additional Experimental Details}

\subsection{Dynamical Systems and Multi-task Settings}
The dynamical systems considered in this work are inherently continuous in time. However, many real-world systems are implemented in discrete time using numerical integration techniques. Following this approach, we generate discrete-time trajectories using the fourth-order Runge-Kutta (RK4) integration method.
\subsubsection{Simple Pendulum}
The 1-link pendulum, also known as the simple pendulum, is a benchmark autonomous system in control theory. It consists of a point mass $m$ suspended from a massless rod of length $l$, attached to a hinge with coefficient of friction $b$. This system is stable around the vertically downward equilibrium position. The states are angular displacement from the equilibrium position $\theta$, and the angular velocity $ \dot \theta$. Given any initial angular displacement and angular velocity, the system eventually returns to its equilibrium position due to friction-based damping. The system dynamics are described as follows:
\begin{subequations}
    \begin{align}
        \dfrac{d x_1}{dt} & = x_2, \\
        \dfrac{d x_2}{dt} &= -\dfrac{g}{l}\sin{x_1}-\dfrac{b}{ml^2} x_2,
    \end{align}
\end{subequations}
where $x_1 = \theta$ and $x_2 = \dot{\theta}$ represent the angular position and angular velocity, respectively; and $g$ is the gravitational constant.

To create a diverse pool of multi-task trajectories, we model $(g, l, b)$ to be stochastic variables following a Gaussian distribution $\mathcal{N}(\mu_\text{sp}, \Sigma_\text{sp})$, where $\mu_\text{sp}$ and $\Sigma_\text{sp}$ denote the mean and covariance matrix, respectively. The shorthand subscript $\text{sp}$ represents the single-link pendulum system. Finally, we set $\mu_\text{sp} = [9.8, 3, 0.5]$ and $\Sigma_\text{sp} = \text{diag}([1, 1, 0.01])$, and draw 60 random parameter samples where each sample realizes a distinct dynamical system.
\subsubsection{Double Pendulum}
The double (2-link) pendulum, also referred to as a chaotic pendulum, consists of two point masses, $m_1$ and $m_2$, connected by a massless rod of length $l_2$. Similar to a simple pendulum, mass $m_1$ is suspended from a fixed hinge via another massless rod of length $l_1$, with a friction coefficient $b$ at the hinge. This configuration forms a coupled nonlinear system that exhibits chaotic behavior. The dynamics of the double pendulum are described by the following set of ordinary differential equations:
\begin{subequations}
    \begin{align}
        \dfrac{d x_1}{dt} & = x_2, \\
        \dfrac{d x_2}{dt} & = 
        \frac{
            -g(2m_1 + m_2)\sin x_1 
            - m_2 g \sin(x_1 - 2x_3)
            - 2\sin(x_1 - x_3) m_2 \left(x_4^2 l_2 + x_2^2 l_1 \cos(x_1 - x_3)\right)
        }{
            l_1 \left(2m_1 + m_2 - m_2 \cos(2x_1 - 2x_3)\right)
        } -b_1x_2, \\
        \dfrac{d x_3}{dt} & = x_4, \\
        \dfrac{d x_4}{dt} & =
        \frac{
            2 \sin(x_1 - x_3) \left(x_2^2 l_1 (m_1 + m_2) + g (m_1 + m_2) \cos x_1 + x_4^2 l_2 m_2 \cos(x_1 - x_3) \right)
        }{
            l_2 \left(2m_1 + m_2 - m_2 \cos(2x_1 - 2x_3)\right)
        } -b_2 x_4,
    \end{align}
\end{subequations}
where $x_1 = \theta_1$,  $x_2 = \dot{\theta_1}$, $x_3 = \theta_2$ and  $x_4 = \dot{\theta_2}$ represent the angular positions and velocities of masses $m_1$ and $m_2$ respectively, and $g$ is the gravitational constant.

To construct a multi-task data pool, we assume $(g, l_1, l_2, m_1, m_2, b_1, b_2)$ to be stochastic parameters following a Gaussian distribution $\mathcal{N}(\mu_\text{dp}, \Sigma_\text{dp})$, where $\mu_\text{dp}$ and $\Sigma_\text{dp}$ denote the mean and covariance matrix, respectively. The subscript $\text{dp}$ signifies the double pendulum system. In the end, we set $\mu_\text{dp} = [9.8, 3, 3, 1, 1, 0.5, 0.5]$ and $\Sigma_\text{sp} = \text{diag}([0.5, 0.5, 0.5, 0.2, 0.2, 0.01, 0.01])$, and draw 120 random parameter samples where each sample defines a unique dynamical system instance.

\subsubsection{Power System Network Containing Five Microgrids}
Microgrids are localized electrical systems designed to supply power to critical infrastructure, such as hospitals. These systems can operate both in grid-connected mode by interfacing with the larger distribution grid and in islanded mode, functioning autonomously. We consider a network comprising five microgrid units and operating in grid-connected mode via points of common coupling (PCCs). In this setting, all microgrids are individually equipped with an angle-droop controller, and the network follows AC power flow dynamics. Assuming a time-scale separation, we treat voltage magnitudes as constant, which reduces the system to the following five-dimensional dynamical model:
\begin{subequations}
    \begin{align}
    J_{\delta_i} \Delta \dot{\delta}_i=&-D_{\delta_i} \Delta \delta_i-\Delta P_i \\
    P_i=&\sum_{k \neq i} E_i^* E_k^* Y_{i k} \cos \left(\delta_{i k}-\gamma_{i       k}\right)
    +E_i^{* 2} G_{i i}, \hspace{3em}i \in\{1, \ldots, 5\},
    \end{align}
\end{subequations}
where  $\Delta \delta_i =\delta_i-\delta_i^*$; $\Delta P_i =P_i-P_i^*$, $\delta_i^*$, $P_i^*$, $Q_i^*$ represent the nominal set point values of the voltage magnitude, phase angle, and active power at the $i$\text{th} microgrid, respectively; $J_{\delta_i}$ is the $i$-th tracking time constant; $D_{\delta_i}$ is $i$-th droop coefficient; $Y_{ik}$ and $G_{ii}$ are Admittance matrix parameters. while $E_i^* = 1.0~ \text{p.u.}$ is the nominal voltage magnitude.

For a multi-task data pool, we treat the droop coefficients $(D_{\delta_1}, D_{\delta_2}, D_{\delta_3}, D_{\delta_4}, D_{\delta_5})$ as stochastic parameters following a Gaussian distribution $\mathcal{N}(\mu_\text{mg}, \Sigma_\text{mg})$, where $\mu_\text{mg}$ and $\Sigma_\text{mg}$ denote the mean and covariance matrix, respectively. The subscript $\text{mg}$ signifies the microgrid network. We select $\mu_\text{mg} = [2.0, 2.0, 2.0, 2.0, 2.0]$ and $\Sigma_\text{mg} = \text{diag}([0.2, 0.2, 0.2, 0.2, 0.2])$, and generate 100 random parameter samples where each sample instantiates a unique system configuration. We note that, remaining system parameters are constant with values adopted from \cite{jena2022fulldistnlf}.
\subsubsection{SEIR Epidemic Model of Two Interacting Populations}
The SEIR model is a mathematical framework widely used in epidemiology to understand and predict the dynamics of an infectious disease outbreak \cite{arino2007seir, turkyilmazoglu2024SEIR}. It splits a general population into four groups: susceptible (S), exposed (E), infectious (I), and removed (R). Each group is described as follows:
\begin{itemize}
    \item \textbf{Susceptible}: Individuals who are at risk of contracting the disease.
    \item \textbf{Exposed}: Individuals who have already contracted the disease, but are not yet infectious themselves.
    \item \textbf{Infectious}: Individuals who are already infectious and can spread the disease to others.
    \item \textbf{Removed}: Individuals who can no longer contract the disease due to recovery, immunity, demise, or effective isolation.
\end{itemize}
We consider a setting where two interacting populations, such as residents in two neighboring cities, and formulate a coupled SEIR model to capture the cross-population dynamics. The following system of ordinary differential equations characterizes the model.
\begin{subequations}
\begin{align}
\frac{dS_1}{dt} &= -\beta_{11} S_1 I_1 - \beta_{12} S_1 I_2 -v_1S_1, \\
\frac{dE_1}{dt} &= \beta_{11} S_1 I_1 + \beta_{12} S_1 I_2 - \sigma E_1, \\
\frac{dI_1}{dt} &= \sigma E_1 - \gamma_1 I_1, \\
\frac{dR_1}{dt} &= \gamma_1 I_1 + v_1S_1, \\[1em]
\frac{dS_2}{dt} &= -\beta_{22} S_2 I_2 -  \beta_{21} S_2 I_1 v_2S_2, \\
\frac{dE_2}{dt} &= \beta_{22} S_2 I_2 + \beta_{21} S_2 I_1 - \sigma E_2, \\
\frac{dI_2}{dt} &= \sigma E_2 - \gamma_2 I_2, \\
\frac{dR_2}{dt} &= \gamma_2 I_2+ v_2S_2,
\end{align}
\end{subequations}

where, 
\begin{itemize}
    \item $S_i, E_i, I_i, R_i$ represent the number of susceptible, exposed, infectious, and removed individuals in population $i \in \{1, 2\}$. We assume that each population size $N_i$ is constant and very large, and that the compartment variables $S_i$, $E_i$, $I_i$, and $R_i$ are normalized by their respective population sizes (i.e., they represent fractions of the total population).
    \item $\beta_{ij}$ denotes the transmission rate from infectious individuals in population $j$ to susceptible individuals in population $i$,
    \item $\sigma$ is the rate at which exposed individuals become infectious (i.e., the inverse of the incubation period).
    \item $\gamma_1$ and $\gamma_2$ are the recovery or removal rates of population $i \in \{\ 1, 2\}$ (i.e., the inverse of the infectious period).
    \item $v_1$ and $v_2$ are vaccination rates that steer the susceptible individuals towards the removed class due to immunity to the disease. 
\end{itemize}
The above system is stable if $S_i = \dfrac{\beta_{ii}}{\gamma_i}$, where $i \in \{1, 2\}$. In the long run, the susceptible, exposed, and infectious compartments ($S_i$, $E_i$, $I_i$) tend to zero, while the removed compartment $R_i$ approaches the total population size $N_i$. Since we normalize all variables by $N_i$, the removed fraction $R_i$ tends to 1, and we effectively track deviations from this limit (e.g., $1 - R_i$). We make a change of variable by introducing  $\tilde R_i = R_i-1$ for $i\in \{1,2\}$ to make the equilibrium coincide with the origin. The shifted dynamics are 
\begin{align*}
\frac{dS_1}{dt} &= -\beta_{11} S_1 I_1 - \beta_{12} S_1 I_2 - v_1S_1 \\
\frac{dE_1}{dt} &= \beta_{11} S_1 I_1 + \beta_{12} S_1 I_2 - \sigma E_1 \\
\frac{dI_1}{dt} &= \sigma E_1 - \gamma_1 I_1 \\
\frac{d\tilde{R}_1}{dt} &= \gamma_1 I_1 + v_1S_1\\
\\
\frac{dS_2}{dt} &= -\beta_{22} S_2 I_2 - \beta_{21} S_2 I_1 - v_2S_2 \\
\frac{dE_2}{dt} &= \beta_{22} S_2 I_2 + \beta_{21} S_2 I_1 - \sigma E_2 \\
\frac{dI_2}{dt} &= \sigma E_2 - \gamma_2 I_2 \\
\frac{d\tilde{R}_2}{dt} &= \gamma_2 I_2 + v_2S_2
\end{align*}
The above dynamics have the exact same form as the original dynamics. However, we emphasize that instead of tracking $R_i$, we track $\tilde R_i$, which eventually leads to a stable equilibrium at the origin $(0, 0, 0, 0, 0, 0, 0, 0).$

We assume self-transmission and recovery rates $(\beta_{11}, \beta_{22}, \gamma_1, \gamma_2)$ as stochastic parameters jointly governed by $\mathcal{N}(\mu_{\text{SEIR}}, \Sigma_{\text{SEIR}})$.
We choose $\mu_\text{SEIR} = [0.2, 0.2, 0.75, 0.75]$ and $\Sigma_{\text{SEIR}} = \text{diag}([0.007, 0.007, 0.007, 0.007])$, and draw 120 samples for the multi-task pool. For all other parameters, we set $\beta_{1,2} = \beta{2,1} = 0.01$, $\sigma = 0.2$.
\subsection{Discretized Reaction-Diffusion PDE as a High-Dimensional Dynamical System}
Reaction-diffusion PDEs describe the spatiotemporal evolution of a variable $U$ across a spatial domain. Such PDEs are generally represented as follows:
\begin{equation}
    \label{eq:PDE}
    \dfrac{\partial U}{\partial t} = \alpha_\text{diff} \nabla_{x} U + f^{r}(U),
\end{equation}
  where $f^{r}(.)$ denotes the reaction term, which characterizes the underlying physical process, and $\alpha_\text{diff}$ is the diffusion coefficient regulating the rate and spatial spread of $U$. We choose the variable $U(x)$ to represent the temperature at a spatial point $x$. For the reaction term, we consider $-x(1+x^2)$, which makes the variable $U$ stabilize over time and asymptotically converge to the origin. Given a spatial domain $X$ of size $[0,10]^2$, using $\Delta x = 1 $ unit,  we divide $X$ into a $10\times10$ grid,  obtaining 100 nodal points as follows: 
  \begin{equation*}
      x^{\Delta} =[x_{1,1}, x_{1,2}, \hdots, x_{1, 10}, \hdots, x_{2,1}, x_{2,2}, \hdots, x_{2, 10}, \hdots, x_{10, 1}, x_{10, 2}, \hdots, x_{10, 10}],
  \end{equation*}
Using a finite-difference spatial discretization, we approximate the spatial Laplacian operator present in 
\eqref{eq:PDE}, and construct a 100-dimensional system of ODEs. The temperature at nodal points is decribed ny $U^\Delta = [u_{1,1}, u_{1,2}, \hdots, u_{1, 10}, \hdots, u_{2,1}, u_{2,2}, \hdots, u_{2, 10}, \hdots, u_{10, 1}, u_{10, 2}, \hdots, u_{10, 10}]$. Now, we formally state the dynamical system as follows: 
\begin{equation}
   U^{\Delta}_{k+1} = U^{\Delta}_k + \Delta t \left( \alpha_\text{diff} \mathbf{L} U^{\Delta}_k + f^r(U^{\Delta}_k) \right),
\end{equation}
where $\textbf{L}$ is $100\times100$ discrete Laplacian matrix, $f^r(U^{\Delta}_k)$ is element-wise computed reaction term and $\Delta t$ is the time-step size. We set $\Delta t = 0.001$. 

In this setting, the diffusion coefficient $\alpha_\text{diff}$ is assumed to be stochastic governed by $\mathcal{N}(\mu_\text{PDE-SM}, \Sigma_\text{PDE-SM})$. We set $\mu_\text{PDE-SM} = [1.2]$ and $\Sigma_\text{PDE-SM} = \text{diag}([0.3])$, and draw 100 random samples to create the multi-task pool.
\subsection{Baselines}
\begin{itemize}
    \item \textbf{In-context Learning (ICL)}: This is a standard ICL approach without any specific emphasis on stability. Using the notations introduced in \textit{LILAD framework Set-up}, the ICL training in our case aims to optimize the following loss function:
    \begin{equation}
        \dfrac{1}{M(n+1)} \sum_{i=1}^M \sum_{j=0}^n  \|G_\theta(x_{i, j+1}|\mathscr{C}_j^i)
         - f_{\tilde\vartheta_i}(x_{i, j+1})\|,
    \end{equation}
    where $G_\theta(.)$ is the ICL-based dynamics model, $\mathscr{C}_j^i$ is the $(i,j)$-th context prefix as defined in the main paper, and $M$ and $n$ are the number of tasks and number of samples per task, respectively. Given a new test-time context prefix $\mathscr{C}^{M+1}_j$, a fully-trained ICL-model $G_{\theta^{*}}$ makes a prediction as $ \hat{x}_{k+1} = G_{\theta^{*}}(x_k | \mathscr{C}_j^{M+1})$.
    \item \textbf{Model Agnostic Meta-learning (MAML)}:
    We use a one-step-based MAML in our case. In this setting, we divide the data pool $\mathcal{D}^\text{train} := \{\mathcal{D}^{i}\}_{i=1}^{M} = \{\{(x_{i,j}, f_{\tilde \vartheta_i}(x_{i,j}))\}_{j=1}^n\}_{i=1}^M$ into $\{\mathcal{D}_{i}^\text{MAML-tr}, \mathcal{D}_{i}^\text{MAML-te}\}_{i=1}^M$, where $\mathcal{D}_{i}^\text{MAML-tr}$ and $ \mathcal{D}_{i}^\text{MAML-te}$ denote the training and test batches containing $p$ and $n-p$ samples, respectively., employed for inner-loop adaptation and outer-loop evaluation. The MAML training objective minimizes the following meta-loss:
    \begin{equation}
        \dfrac{1}{M(n-p)}\sum_{i=1}^{M}\sum_{j=1}^{n-p} L\Big( M_{\theta - \tilde\alpha\nabla L(\mathcal{D}_{i}^\text{MAML-tr})}(x_{i,j}), f_{\tilde \vartheta_{i}}(x_{i,j})\Big),
    \end{equation}
    where $L(x, y) = \|x-y\|$, and $L(\mathcal{D}) = \dfrac{1}{|\mathcal{D}|}\sum_{(x,y) \in \mathcal{D}}\|x-y\|$, $M_\theta$ is the meta-model and $\tilde \alpha$ is the inner-loop learning rate. During evaluation, using a new test-time dataset $\mathcal{D}_{M+1}^\text{MAML-tr}$ (processed from the test-time prompt $\mathscr{C}_j^{M+1}$), the gradient-based adaptation is performed as $\theta_{M+1}^{*} = \theta^{*} - \tilde\alpha L(\mathcal{D}_{M+1}^\text{MAML-tr})$, where $\theta^{*}$ and $\theta^{*}_{M+1}$ are optimal meta-parameter (from MAML training), and optimal adapted task parameter for $(M+1)$-th task, respectively.
    \item \textbf{Conditional Value at Risk (CVaR)}: The CVaR is a risk-averse technique that computes the expected loss across the worst-case scenarios beyond a certain confidence level. Using $\mathcal{D}^\text{CVaR} := \{\{(x_{i,j}, f_{\tilde \vartheta_i}(x_{i,j}))\}_{j=1}^n\}_{i=1}^M$, CVaR optimizes the following nested CVaR objective:
    \begin{equation}
        \min_{\tau_1, \tau_2 \in \mathbb{R}} \left\{ \tau_2 + \frac{1}{\beta Mn} \sum_{i=1}^{M} \left[\tau_1+ \sum_{j=1}^{n}  \left[ \|C_\theta(x_{i,j}) - f_{\tilde \vartheta_i}(x_{i,j})\| - \tau_1 \right]_+ -\tau_2\right]_+ \right\},
    \end{equation}
    where \( [a]_+ = \max(a, 0) \), $\beta \in (0,1)$ is the confidence level, and $\tau_1, \tau_2$ are auxiliary variables that approximate sample-level and task-level  value-at-risk, respectively. 
    \item \textbf{Stable-linear}:
    This method is non-adaptive and computes a linear matrix $A^{M+1}$ given the test-time dataset  $\mathcal{D}^{M+1} = \{(x_{M+1,j}, f_{\tilde \vartheta_{M+1}}(x_{M+1,j}))\}_{j=1}^n$. The stable-linear utilizes a convex formulation-based technique to minimize the following objective function:
    \begin{equation}
        \min_{\theta} \sum_{j=1}^{n-1} \left\| A^{M+1}x_{M+1,j} - f_{\tilde \vartheta_{M+1}}(x_{M+1,j}) \right\|^2 \quad \text{subject to} \quad \|A^{M+1}\| < 1,
    \end{equation}
    where the constraint ensures that the eigenvalues of $A^\text{M+1}$ lie in the unit circle, establishing the stability of the linear discrete-time dynamical system.
\end{itemize}
\subsection{LILAD Hyperparameters}
For LILAD, we provide the necessary hyperparameters in the following table:
\begin{table*}[h]
\caption{LILAD Hyperparameters Across Environments}
\label{tab:hyperparam}
\begin{center}
\begin{small}
\begin{sc}
\begin{tabular}{lclccccr}
\hline
Env. & GPT-2 Size & Tr. Epochs & Lr$_1$ & Lr$_2$ & Tr. samples per task  & Tasks & Context Size\\
\hline
SP
&(128, 6, 4) & 300K  & $5*10^{-7}$  & $1*10^{-8}$ & 100K & 100  & 128  \\
\hline
DP
& (128, 6, 4) & 500K & $1*10^{-8}$  & $1*10^{-8}$ & 120K & 120  & 128  \\
\hline
MG
& (128, 6, 4)&  400K &  $1*10^{-8}$ &$1*10^{-8}$   & 100K & 100  & 128  \\
\hline
SEIR
& (128, 6, 4) &  800K& $1*10^{-8}$  & $1*10^{-8}$  & 150K & 120  & 128 \\
\hline
PDE-SM 
& (512, 12, 8)& 2M & $1*10^{-8}$   & $1*10^{-8}$   & 200K & 60  &256  \\
\hline
\end{tabular}
\end{sc}
\end{small}
\end{center}
\vspace{-1em}
\end{table*}

In this table GPT-2 Size indicates (embedding dimension, no. of transformer blocks, no. of attention heads in each self-attention layer); Tr. Epochs is the number of training epochs for both model; $Lr_1$ and $Lr_2$ are learning rates for the model and Lyapunov function, respectively; Tr. samples per task is the number of training samples per task; Tasks is total number training tasks considered; and context size is the size of the context (prompt excluding the query) during training and testing for both models. The exponential stability coefficient $\beta$ for Lyapunov function is set as $0.95$ for all cases.
\end{document}